\documentclass[twocolumn,prb,floatfix,nofootinbib,superscriptaddress]{revtex4-2}

\usepackage[utf8]{inputenc}
\usepackage[T1]{fontenc}
\usepackage[english]{babel}
\usepackage{amsmath,amssymb,amsthm}
\usepackage[normalem]{ulem}
\usepackage{xcolor}
\usepackage{graphicx}
\usepackage{wasysym}
\usepackage{physics}

\newcommand{\Nb}{N_{\mathrm b}}
\newcommand{\Nf}{N_{\mathrm f}}
\newcommand{\T}[1]{\overline{#1}}             
\newtheorem{proposition}{Proposition}

\usepackage{hyperref}
\hypersetup{
    colorlinks,
    citecolor=black,
    filecolor=black,
    linkcolor=black,
    urlcolor=blue
}

\begin{document}

\title{Exact SU($N$) symmetry resolution of bosonic Hubbard models: general boson-fermion mappings and flat-band singlet-cluster states}
\author{Lo\"ic Herviou}
\affiliation{Univ. Grenoble Alpes, CNRS, LPMMC, 38000 Grenoble, France}

\begin{abstract}
We show how to perform exact diagonalization of SU($N$)-symmetric Bose-Hubbard models separately in each irreducible representation of SU($N$), by exploiting the Howe duality between U($L$) and U($N$).
Three structural consequences follow for any lattice and any Hamiltonian built from the U($L$) generators: (i) the spectrum in a given orbital sector is independent of $N$ and of the particle statistics, (ii) the canonical partition functions for bosons and fermions obey the exact identity $Z_{\mathrm b}(-N)=(-1)^{\Nb}Z_{\mathrm f}(N)$, (iii) with unfrustrated hopping the groundstate is always a color ferromagnet, so that the color-resolved method we introduce is useful to study excited sectors, frustrated lattices and finite temperatures.
As applications, we first derive the finite-temperature phase diagram of attractive SU($N$) bosons on few-site rings at large $N$ for arbitrary filling.
We then show the emergence of singlet-cluster states from frustrated hopping, the counterpart of fermionic ferromagnetism, with the example of the sawtooth and kagome lattice. 
\end{abstract}

\maketitle

\section{Introduction}
Ultracold alkaline-earth(-like) atoms have turned SU($N$)-symmetric Hubbard models from a theoretical device into an experimental platform~\cite{Gorshkov2010,CazalillaRey2014}.
Because their nuclear spin decouples from the electronic degrees of freedom, fermionic isotopes of ytterbium and strontium realize flavor-independent interactions with up to $N=2I+1=10$ components.
This has enabled the observation of SU(6) Mott insulators~\cite{Taie2012}, one-dimensional liquids with tunable $N$~\cite{Pagano2014}, SU($N$)-symmetric spin exchange~\cite{Scazza2014,Zhang2014} and antiferromagnetic correlations in SU($N$) Hubbard models~\cite{Hofrichter2016,Ozawa2018,Taie2022}, with quantum-gas microscopes now resolving multi-flavor systems site by site~\cite{Mongkolkiattichai2025}.
%
A striking feature of these experiments is that, as $N$ grows at fixed particle number, SU($N$) fermions increasingly behave like bosons~\cite{Pagano2014,Song2020}.

Bosons with SU($N$)-symmetric interactions have long remained out of reach, since the alkaline-earth mechanism mainly provides SU($N$) symmetry for fermions with repulsive interactions~\cite{Mukherjee2025}.
Spinor Bose gases offer only approximate realizations, with small spin-dependent scattering lengths~\cite{StamperKurnUeda2013}.
This has recently changed: ultracold polar molecules shielded from inelastic collisions by static or microwave fields are predicted to display SU($N$) symmetry for both bosons and fermions, with attractive or repulsive interactions and potentially large $N$~\cite{Mukherjee2025,MukherjeeHutson2025}.
In parallel, bosonic molecules have been brought to quantum degeneracy, up to Bose-Einstein condensation~\cite{Lin2023,Bigagli2024}.
SU($N$) Bose-Hubbard models are thus becoming experimentally relevant, in different regimes than their fermionic alkaline-earth counterparts.

For bosons with unfrustrated hopping, however, SU($N$) symmetry is not expected to produce new groundstate physics.
The Eisenberg-Lieb theorem~\cite{Eisenberg2002} ensures that the groundstate is an SU($N$) ferromagnet, whose energy is that of spinless bosons: consequently, color physics requires frustration.
Frustrated bosonic hopping is available in cold-atom experiments: periodic driving can reverse the sign of the tunneling amplitude~\cite{Eckardt2005,Lignier2007,Eckardt2017}, which was used to realize frustrated magnetism of bosons on the triangular lattice~\cite{Struck2011,Struck2013}.
Optical lattices with flat bands, such as the kagome~\cite{Jo2012} and Lieb~\cite{Taie2015} lattices, have been implemented, and negative-temperature states~\cite{Braun2013} can effectively turn the upper flat band of the Kagome lattice into its lowest one.
Flat bands are also realized in photonic and polaritonic systems~\cite{Leykam2018}.
For spinless bosons, frustrated flat bands lead to crystals of localized states~\cite{Schulenburg2002,Zhitomirsky2004,Huber2010,Derzhko2015}, while for SU($N$) fermions they underlie flat-band ferromagnetism~\cite{Mielke_1991a,Tasaki1992,LiuNieZhang2019,
jin2026flatbandferromagnetismsunhubbard}.
Their SU($N$) bosonic counterpart has, to our knowledge, not been studied.

Addressing these questions numerically requires resolving the full SU($N$) symmetry, since the Hilbert space grows as $N^{\Nb}$.
For fermions, this has been achieved with bases adapted to the irreducible representations of SU($N$): standard Young tableaux for Heisenberg models~\cite{NatafMila2014}, and semi-standard Young tableaux of the unitary group U($L$) for Fermi-Hubbard models~\cite{Botzung2024,BotzungNagaoka2024}, extended to DMRG~\cite{Nataf2024DMRG}.
The latter approach also gives access to few-site systems in the limit of large $N$, relevant for double wells~\cite{Albiez2005}, tweezer arrays and molecular platforms, where it revealed zero- and finite-temperature symmetry-breaking transitions~\cite{Nataf2025,Herviou2026}.
In all these works the Pauli principle is essential: it ties the orbital irrep to the \emph{transpose} of the color irrep.
For bosons, numerical studies have so far relied at best on the abelian U(1)$^N$ subgroup.

In this work, we show that the fermionic color-factorization scheme carries over to bosons, with one modification: the orbital and color irreps coincide.
Our main results are:
(i) an exact diagonalization scheme resolving the full SU($N$) symmetry of Bose-Hubbard models, including soft- and hard-core constraints (Sec.~\ref{sec:method});
(ii) a dictionary between bosonic and fermionic spectra valid on any lattice with an exact identity $Z_{\mathrm b}(-N)=(-1)^{\Nb}Z_{\mathrm f}(N)$ between canonical partition functions (Sec.~\ref{sec:dualities});
(iii) the finite-temperature transitions of attractive SU($N$) bosons on few-site rings at large $N$ and arbitrary filling, from which the fermionic results of Ref.~\cite{Herviou2026} follow by $N\to-N$ (Sec.~\ref{sec:0D});
(iv) \emph{flat-band singlet-cluster states}, the bosonic counterpart of flat-band ferromagnetism (Sec.~\ref{sec:Frustration}).
We prove that frustrated flat bands stabilize color-antisymmetric groundstates for $\Nb\le N$ and construct exact SU($N$)-singlet groundstates at finite density on the sawtooth chain and the kagome lattice, where their existence reduces to a packing problem.

\section{color factorization for bosons}\label{sec:method}
\paragraph{Model.}
We consider SU($N$) Hubbard models on a lattice of $L$ sites,
\begin{equation}
    H = -\sum_{i, j} t_{ij}\, E_{i, j} + \frac{U}{2} \sum_{i} E_{i, i}^2,
    \qquad
    E_{i, j} = \sum_{\sigma=1}^N c^\dagger_{i, \sigma} c_{j, \sigma},
    \label{eq:Ham}
\end{equation}
where $t_{ji}=t_{ij}^*$ are hopping amplitudes, $c_{i,\sigma}$ are bosonic (or fermionic)
annihilation operators, and $\sigma\in\{1,\dots,N\}$ labels the color.
We denote by $n_i=E_{i,i}=\sum_\sigma n_{i,\sigma}$, with $n_{i,\sigma}=c^\dagger_{i,\sigma}c_{i,\sigma}$, the number of particles on site $i$.
The interaction $E_{i,i}^2=n_i^2$ differs from the conventional $n_i(n_i-1)$ by $n_i$, though at fixed total particle number $\Nb$ and uniform $U$, the two conventions give identical spectra (and eigenstates) up to the shift $-U\Nb/2$.
For either statistics, the $E_{i,j}$ satisfy the commutation relations of $\mathfrak{gl}(L)$, the complexification of $\mathfrak u(L)$,
\begin{equation}
    [E_{i, j}, E_{k, l}] = \delta_{j, k} E_{i, l} - \delta_{i, l} E_{k, j}.
    \label{eq:uL}
\end{equation}
The rest of this Section holds for \emph{any} Hamiltonian that is a polynomial in the $E_{i,j}$: arbitrary hoppings, including complex ones, density-density interactions $V_{ij}E_{i,i}E_{j,j}$, pair hopping, exchange terms $E_{i,j}E_{j,i}$\dots
It does not cover color-dependent interactions, such as $\sum_{i,\sigma}n_{i,\sigma}^2$, which break the SU($N$) symmetry, or particle-number non-conserving terms.
The use of the generators $E_{i,j}$ and of Gelfand-Tsetlin bases for many-body problems goes back to the unitary group approach of nuclear physics and quantum
chemistry~\cite{Moshinsky1968,Paldus1974,Shavitt1977}; for SU($N$) lattice models it was developed in Refs.~\cite{NatafMila2014,Botzung2024,Nataf2024DMRG}.\\

\paragraph{Young diagrams and tableaux.}
A Young diagram $\alpha=(\alpha_1\ge\alpha_2\ge\dots\ge\alpha_{\ell(\alpha)}>0)$ with $|\alpha|=\sum_k\alpha_k=n$ boxes is a left-justified array of $n$ boxes in $\ell(\alpha)$ rows, the $k$-th row containing $\alpha_k$ boxes.
$\T\alpha$ denotes the transposed diagram.
Diagrams with $\ell(\alpha)\le N$ label the irreducible representations (irreps) of U($N$) appearing in $(\mathbb{C}^N)^{\otimes n}$~\cite{Weyl1939,FultonHarris1991}.
In the Young-symmetrizer construction, rows symmetrize and columns antisymmetrize the tensor indices; the single row $(n)$ is the fully symmetric irrep and the single column $(1^n)$, with $n\le N$, the fully antisymmetric one.
A basis of the irrep $\alpha$ is labeled by the semi-standard Young tableaux (SSYT) of shape $\alpha$~\cite{Fulton1997}: fillings of the boxes with labels in $\{1,\dots,N\}$, weakly increasing along rows and strictly increasing down columns.
This basis is in one-to-one correspondence with Gelfand-Tsetlin patterns~\cite{GelfandTsetlin1950,Alex2011}.
Each basis vector is a weight vector with weight component $m_i$ the number of occurrences of the label $i$, i.e., an SSYT is an eigenvector of $E_{j, j}$ with eigenvalue $m_j$.
The dimension of the irrep is the number of SSYT, given by the hook-content
formula~\cite{Stanley1999}
\begin{equation}
  d_{\alpha,N}=\prod_{b\in\alpha}\frac{N+c(b)}{h(b)},
  \label{eq:hookcontent}
\end{equation}
where a box $b$ in row $r$ and column $s$ has content $c(b)=s-r$ and hook length $h(b)$, the number of boxes to its right, below it, and itself.\\

\paragraph{Howe duality.}
On the $\Nb$-particle Fock space, the operators $E_{i,j}$ (which generate U($L$), due to Hermiticity) and the color rotations (U($N$)) generate mutually commuting algebras, each the commutant of the other.
For bosons, the Hilbert space $\mathcal H_{\rm b}=\mathrm{Sym}^{\Nb}(\mathbb C^L\otimes\mathbb C^N)$ decomposes as~\cite{Howe1989,Howe1995,GoodmanWallach2009, Neergaard2021}
\begin{equation}
    \mathcal{H}_{\rm b}=\bigoplus_{\substack{|\alpha| = \Nb\\ \ell(\alpha)\leq \min(L,N)}}
    V^{(L)}_{\alpha}\otimes V^{(N)}_{\alpha},
    \label{eq:howe_b}
\end{equation}
where U($L$) acts on $V^{(L)}_{\alpha}$ and U($N$) on $V^{(N)}_{\alpha}$ as the irrep $\alpha$.
For $N_{\rm f}$ fermions~\cite{Howe1995,GoodmanWallach2009,Neergaard2021,Botzung2024}, $\mathcal H_{\rm f}=\wedge^{N_{\rm f}}(\mathbb C^L\otimes\mathbb C^N)$ and
\begin{equation}
    \mathcal{H}_{\rm f}=\bigoplus_{\substack{|\alpha| = N_{\rm f}\\ \ell(\alpha)\leq L,\ \alpha_1\leq N}}
    V^{(L)}_{\alpha}\otimes V^{(N)}_{\T\alpha}.
    \label{eq:howe_f}
\end{equation}
The origin of the difference is the Schur-Weyl duality: if the orbital and color parts of an $n$-particle wavefunction transform under irreps $\alpha$ and $\mu$ of the symmetric group $S_n$, $\alpha\otimes\mu$ contains the trivial representation (bosons) only if $\mu=\alpha$, and the sign representation (fermions) only if $\mu=\T\alpha$.

Hence, in the bosonic sector of color irrep $\alpha$,
\begin{equation}
  \dim\mathcal{H}_{{\rm b}, \alpha} = d_{\alpha,L}d_{\alpha,N},
  \qquad
  H\big|_{\mathcal{H}_{{\rm b},\alpha}}= H_\alpha\otimes \mathrm{Id}_{d_{\alpha,N}},
  \label{eq:dimsector}
\end{equation}
where $H_\alpha$ is the $d_{\alpha,L}\times d_{\alpha,L}$ matrix of $H$ in the U($L$) irrep $\alpha$.
We use as orthonormal basis of $V^{(L)}_{\alpha}$ the SSYT of shape $\alpha$ filled with site labels.
The action of the $E_{i,j}$ on this basis depends only on the U($L$) irrep and not on the statistics: it is given by the formulas of Ref.~\cite{Botzung2024}, derived there for the fermionic orbital shapes $\T\alpha$ and applied here to the shapes $\alpha$.
This yields $H_\alpha$ directly for any irrep $\alpha$, free of the SU($N$) degeneracy.\\

\paragraph{Constrained bosons.}
The basis states are eigenstates of the $E_{j,j}$: a tableau in which the site label $j$ appears $k$ times has $k$ bosons at position $j$.
The projector $P_n$ onto states with at most $n$ bosons per site is therefore diagonal in this basis, and the constrained model $P_nHP_n$ is obtained by discarding the tableaux in which some label appears more than $n$ times.
For $n=1$, only tableaux with distinct entries survive, and, in the Mott limit and at filling $1$, we recover the standard Young tableaux basis of Ref.~\cite{NatafMila2014}.
Since $P_n$ is itself a polynomial in the $n_i$, $P_nHP_n$ in fact belongs to the algebra generated by the $E_{i,j}$, and all the results of this section, and of Sec.~\ref{sec:dualities}, apply to constrained models as well.

\section{Structural consequences}\label{sec:dualities}
The color factorization of Sec.~\ref{sec:method} reduces any SU($N$)-invariant Hamiltonian to a family of orbital problems $H_\alpha$, labeled by U($L$) irreps $\alpha$, which depend neither on $N$ nor on the particle statistics.
In this section, we first derive the resulting exact relations between spectra and between thermodynamic quantities and finally discuss when the color resolution yields a computational gain, and when the groundstate is trivially fixed by symmetry.

\subsection{Spectral relations}\label{sec:spectral}
The following statements hold for any lattice, dimension and Hamiltonian polynomial in the $E_{i,j}$, including the constrained models of Sec.~\ref{sec:method}.\\

\paragraph{$N$-independence.}
The bosonic spectrum in the orbital sector $\alpha$ is identical for all $N\ge\ell(\alpha)$, and the fermionic one for all $N\ge\alpha_1$.
Only the color degeneracies, $d_{\alpha,N}$ for bosons and $d_{\T\alpha,N}$ for fermions, depend on $N$.
A single diagonalization per orbital sector therefore yields the spectrum for all $N$.\\

\paragraph{Boson-fermion correspondence.}
At equal particle number $\Nb=\Nf$, an orbital irrep $\alpha$ is an admissible bosonic sector if $\ell(\alpha)\le\min(N,L)$, and an admissible fermionic sector if $\ell(\alpha)\le L$ and $\alpha_1\le N$.
Comparing these conditions gives three inclusions between the sets of sectors, and hence between the sets of energy levels.\\
(i) If $L\le N$, every fermionic sector is a bosonic sector.\\
(ii) If $\Nb\le N$, every bosonic sector is a fermionic sector.\\
(iii) SU($N$) bosons and SU($N'$) fermions with $N'\ge\Nb$ share exactly the same sectors if and only if $N\ge\min(L,\Nb)$.
The degeneracies of a shared level, $d_{\alpha,N}$ and $d_{\T\alpha,N'}$, differ in general.
The underlying principle, that for spin-independent Hamiltonians the energies depend only on the permutation symmetry of the spatial wave function, with conjugate Young diagrams for bosons and fermions, is familiar from spin-free quantum chemistry~\cite{Matsen1964,Kaplan1975}.\\

\subsection{Thermodynamics: the $N\to-N$ correspondence}\label{sec:NmN}
The canonical partition functions of bosons and fermions at fixed number of particles $N_b$ decompose as
\begin{equation}
Z_{\mathrm b}(N)=\sum_{\alpha} d_{\alpha,N}\,Z_\alpha,\qquad
Z_{\mathrm f}(N)=\sum_{\alpha} d_{\T\alpha,N}\,Z_\alpha,
\end{equation}
with $Z_\alpha=\Tr e^{-\beta H_\alpha}$, and both sums run over all diagrams with $\Nb$ boxes and $\ell(\alpha)\le L$, independently of $N$.
The $N$-dependent admissibility conditions are enforced by the degeneracies: for integer $N\ge 0$, the hook-content formula~\eqref{eq:hookcontent} contains a vanishing factor, so that $d_{\alpha,N}=0$ when $\ell(\alpha)>N$ and $d_{\T\alpha,N}=0$ when $\alpha_1>N$.
Each $d_{\alpha,N}$ is a polynomial of degree $|\alpha|$ in $N$, and so is $Z_{\mathrm b/\mathrm{f} }(N)$.
Transposition flips the contents $c(b)$ while preserving the hook lengths $h(b)$, so that $d_{\alpha,-N}=(-1)^{|\alpha|}d_{\T\alpha,N}$ and
\begin{equation}
  Z_{\mathrm b}(-N)=(-1)^{\Nb}\,Z_{\mathrm f}(N)
  \label{eq:NmN}
\end{equation}
at equal particle number, for any temperature, lattice and Hamiltonian polynomial in the $E_{i,j}$.
The same identity holds for the numerator of any canonical thermal average of an operator in the algebra generated by the $E_{i,j}$, i.e., for all SU($N$)-invariant, particle-conserving operators.
The signs cancel in the ratio, and $\langle O\rangle_{\mathrm b}(-N)=\langle O\rangle_{\mathrm f}(N)$.
In particular, in large-$N$ models at fixed filling $\nu=\Nb/N$, fermionic results follow from bosonic ones by $\nu\to-\nu$ (see Sec.~\ref{sec:0D}).
Equation~\eqref{eq:NmN} is an instance of the negative-dimension correspondence between U($N$) and U($-N$) with transposed Young diagrams~\cite{Parisi1979,Dunne1989,Cvitanovic:2008zz, BereleRegev1987,james2006representation}.

\subsection{When color resolution pays off}\label{sec:gain}
\paragraph{Hilbert-space reduction.}
In the totally symmetric sector $\alpha=(\Nb)$, the color state can be chosen fully polarized, and the resolution of the SU($N$) sectors gives no gain over U(1)$^N$: this sector matches the spinless Bose-Hubbard model.
The gain is in multi-row sectors, in particular for the color singlets.
As shown in Table~\ref{tab:dims}, their orbital dimension is smaller than the dimension of the U(1)$^N$ sector containing them by several orders of magnitude.

\begin{table}[t]
\centering
\begin{tabular}{cccccc}
$\Nb$ & $L$ & $N$ & sector $\alpha$ & $d_{\alpha,L}$ & U(1)$^N$ sector\\ \hline
9 & 9 & any & $(9)$ & $2.43\times10^{4}$ & $2.43\times10^{4}$\\
9 & 9 & 3 & $(3^3)$ & $4.16\times10^{4}$ & $4.49\times10^{6}$\\
9 & 9 & 9 & $(1^9)$ & $1$ & $3.87\times 10^{8}$\\
9 & 12 & any & $(9)$ & $1.68\times10^{5}$ & $1.68\times10^{5}$\\
9 & 12 & 3 & $(3^3)$ & $5.73\times10^{5}$ & $4.82\times10^{7}$\\
9 & 12 & 9 & $(1^9)$ & $220$ & $5.16\times10^{9}$\\
8 & 16 & any & $(8)$ & $4.90\times10^{5}$ & $4.90\times10^{5}$\\
8 & 16 & 2 & $(4^2)$ & $2.37\times10^{6}$ & $1.50\times10^{7}$\\
8 & 16 & 4 & $(2^4)$ & $8.66\times10^{5}$ & $3.42\times10^{8}$\\
8 & 16 & 8 & $(1^8)$ & $1.29\times10^{4}$ & $4.29\times10^{9}$
\end{tabular}
\caption{Dimension $d_{\alpha,L}$ of the orbital problem in the SU($N$) singlet sector (or fully symmetric sector for $\alpha=(\Nb)$), compared with the dimension of the U(1)$^N$ sector containing it, i.e., the sector with equal color populations $\Nb/N$ (fully polarized for $\alpha=(\Nb)$), for bosons. The examples are taken from Sec.~\ref{sec:Frustration}.}
\label{tab:dims}
\end{table}

\paragraph{Unfrustrated bosons: a trivial groundstate.}
For groundstate studies, the reduction is useful only if the groundstate lies outside the fully symmetric sector.
For bosons, this is excluded in a broad class of models.
Suppose that all $t_{ij}$ in Eq.~\eqref{eq:Ham} are real and non-negative, or can be made so by a gauge transformation.
Then all off-diagonal elements of $H$ in the bosonic occupation basis are non-positive.
The Hamiltonian is stoquastic~\cite{Bravyi2008}, hence amenable to sign-free quantum Monte Carlo, and the lattice version of the Eisenberg-Lieb argument~\cite{Suto_1993,Eisenberg2002} shows that a groundstate lies in the totally symmetric sector $\alpha=(\Nb)$, i.e., is an SU($N$) ferromagnet, for any $U$, $N$ and lattice.
If, in addition, the hopping graph is connected and $U<\infty$, the groundstate manifold is exactly this multiplet, and the groundstate problem reduces to spinless bosons.
The color-resolved basis remains useful in this case for excitations in other color sectors (magnons and multi-magnon bound states), at finite temperature, where all sectors contribute, and for dynamics from color-entangled initial states.
The theorem does not apply to frustrated hoppings, i.e., with sign products that cannot be gauged away (e.g. triangular or kagome lattices with the frustrated sign of $t_{ij}$, see Sec.~\ref{sec:Frustration}), nor to generic complex hoppings (e.g. in the presence of a synthetic magnetic field).

\section{Few-site rings at large $N$}\label{sec:0D}
As a first application, we consider few-site systems in the large-$N$ limit.
Following Ref.~\cite{Herviou2026}, we study the ring
\begin{equation}
    H_{0}^{(L)} = -t \sum_{j = 1}^L \left( E_{j, j+1} + E_{j+1, j} \right)
    + \frac{U}{2}\sum_{j = 1}^L  E_{j, j}^2 ,
    \label{eq:Ham0D}
\end{equation}
with $t>0$, $E_{L,L+1}\equiv E_{L,1}$, and consider an unusual thermodynamic limit $N\to\infty$ at fixed $L$ and
fixed filling $\nu=\Nb/N$.
For $L=2$ the periodic sum counts the single bond twice.\\

\paragraph{Zero temperature.}
Since $t>0$, the theorem of Sec.~\ref{sec:dualities} places the groundstate in the sector $(\Nb)$ for every $\nu$. 
It is an SU($N$) ferromagnet, whose multiplet contains the fully polarized states, and its energy is that of $\Nb$ spinless attractive bosons on the ring.
The zero-temperature physics is therefore the self-trapping transition of the attractive Bose-Hubbard ring, which has been studied extensively~\cite{Buonsante2005,Oelkers2007,Javanainen2008,Buonsante2010,Buonsante2012}.
In particular, the critical couplings depend on $U\Nb$ only.
They coincide with those of $\Nf=N$ fermions~\cite{Herviou2026}.
Following the dualities of Sec.~\ref{sec:spectral}, for $\Nf\le N$, the fermionic groundstate also lies in the orbital sector $(\Nf)$, where it carries the color irrep $(1^{\Nf})$ (i.e., the SU($N$) singlet when $\Nf=N$).
For $L=2$, the model reduces in this sector to the Lipkin-Meshkov-Glick model~\cite{Lipkin1965, Nataf2025}.
It has a continuous transition at $U_c\Nb=-4t$, between a uniform phase and a self-trapped phase with a macroscopic population imbalance between the two sites.
For $L\ge3$ the transition becomes first order~\cite{Buonsante2005,Herviou2026}.
The only difference between bosons and fermions is the degeneracy of the groundstate: the bosonic groundstate is the SU($N$) ferromagnet $(\Nb)$, of dimension $\binom{N+\Nb-1}{\Nb}$, instead of the fermionic singlet.\\

\paragraph{Finite temperature, $L=2$.}
At finite temperature, the different color degeneracies of bosons and fermions change the critical temperature.
The bosonic partition function reads
\begin{equation}
    Z_{\mathrm b} = \sum_{\alpha} d_{\alpha, N}\, \Tr\, e^{-\beta H_\alpha},
\end{equation}
where the sum runs over two-row diagrams $\alpha=(\alpha_1,\alpha_2)$.
The orbital dimension $d_{\alpha,2}=\alpha_1-\alpha_2+1$ grows only polynomially, while
\begin{equation}
d_{(\alpha_1,\alpha_2),N}=\binom{N+\alpha_1-1}{\alpha_1}\binom{N+\alpha_2-2}{\alpha_2} \frac{\alpha_1-\alpha_2+1}{\alpha_1+1}
\end{equation}
grows exponentially in $N$.
To leading exponential order we may therefore replace $\Tr e^{-\beta H_\alpha}$ by $e^{-\beta E_{G,\alpha}}$, with $E_{G,\alpha}$ the lowest energy of $H_\alpha$~\cite{Herviou2026}.
We parametrize $\alpha=(\Nb(1-\varepsilon),\Nb\varepsilon)$ with $0\le\varepsilon\le1/2$.
In the strong-coupling, large-$N$ regime $|U|\Nb\gg t$, we can approximate $E_{G,\alpha}\simeq U\Nb^2[1+(1-2\varepsilon)^2]/4$, and $N^{-1}\log d_{\alpha,N}\to f(1+\nu(1-\varepsilon))+f(1+\nu\varepsilon)$ with $f(x)=x\log x-(x-1)\log(x-1)$.
Hence
\begin{multline}
    e^{\beta U \Nb^2/4}\, Z_0 \simeq \int_0^{1/2} d\varepsilon\,
    \exp\Big\{-\frac{\beta U \Nb^2}{4}(1-2\varepsilon)^2 \\
  + N [f(1 + \nu (1-\varepsilon)) + f(1 + \nu \varepsilon)]\Big\}.
\end{multline}
The exponent has two families of stationary points: the symmetric point $\varepsilon=1/2$, stationary at all temperatures, which describes the disordered phase, and the solutions of
\begin{equation}
   \tilde{\beta} \equiv  \beta U \Nb = \frac{1}{1-2\varepsilon}
   \log\!\left[ \frac{1 +\nu(1-\varepsilon)}{1 + \nu \varepsilon}\,  \frac{\varepsilon}{1 -\varepsilon}\right].
    \label{eq:saddleL2}
\end{equation}
The right-hand side increases monotonically from $-\infty$ at $\varepsilon=0$ to $-4/(2+\nu)$ at $\varepsilon=1/2$.
This gives a continuous transition at
\begin{equation}
\beta_c U\Nb=-\frac{4}{2+\nu},
\label{eq:TcL2}
\end{equation}
below which the population imbalance $\langle(n_1-n_2)^2\rangle/\Nb^2\to(1-2\varepsilon^*)^2$ becomes macroscopic, where we denoted by $\varepsilon^*$ the stationary point at effective temperature $\beta$.
Using Eq.~\eqref{eq:NmN}, i.e., setting $\nu=-1$, Eq.~\eqref{eq:TcL2} reproduces the fermionic result $\beta_cUN=-4$ of Ref.~\cite{Herviou2026}.
At equal $N=\Nb$, bosons therefore order at a temperature three times higher than fermions.\\

\paragraph{Finite temperature, $L\ge3$.}
We now write $\alpha_j = \Nb\varepsilon_j$, with $\varepsilon_1=1-\sum_{j\ge2}\varepsilon_j$.
For $U<0$, the lowest interaction energy in the irrep $\alpha$ is reached at the highest weight $n_j=\alpha_j$, so that $E_{G,\alpha}\simeq\frac{U\Nb^2}{2}\sum_{j=1}^L\varepsilon_j^2$ at
strong coupling, and 
\begin{equation}
   \log d_{\alpha,N} \simeq N\sum_{j=1}^L f(1+\nu\varepsilon_{j}).
\end{equation}
On the line $\varepsilon_{j}=\varepsilon$ for $j\ge2$, the saddle-point condition becomes 
\begin{equation}
    \tilde{\beta} = \frac{1}{1 - L\varepsilon}
    \log\!\left[ \frac{1 + \nu (1-(L-1)\varepsilon)}{1 + \nu \varepsilon}
    \,\frac{\varepsilon}{1-(L-1)\varepsilon}\right].
    \label{eq:saddleL}
\end{equation}
We checked numerically, for $3 \leq L \leq 6$, that the global maximum of the exponent lies on this line.
Now, as in the fermionic case, $\tilde{\beta}(\varepsilon)$ is non-monotonic.
Optimizing the free energy ensures that $\varepsilon^*$ jumps discontinuously away from the symmetric point $\varepsilon=1/L$.
The transition is therefore first order, for all values of $\nu$ considered.
The critical values are listed in Table~\ref{tab:CriticalFirstOrder0D}.
The row $\nu=-1$ reproduces the fermionic values of Ref.~\cite{Herviou2026}.
As for $L=2$, these are strong-coupling values, and they are shifted at finite $|U|\Nb/t$.

\begin{table}[t]
    \centering
    \begin{tabular}{c|cccc}
         $L$ & 3 & 4 & 5 & 6  \\ \hline
         $\nu = 2$ & 1.63013 & 2.09970 & 2.47384 & 2.78492 \\
         $\nu = 1$ & 2.04808 & 2.55256 & 2.94481 & 3.26675  \\
         $\nu = 1/2$ & 2.35312 & 2.87070 & 3.26917 & 3.59464 \\
         $\nu = 1/3$ & 2.47721 & 2.99777 & 3.39756 & 3.72376 \\
         $\nu = 1/10$ & 2.67632 & 3.19930 & 3.60014 & 3.92690 \\ \hline
         $\nu=-1$  & 4.4019 & 4.8772 & 5.2756 & 5.6112
    \end{tabular}
    \caption{Strong-coupling critical values of $-\beta_c U\Nb$ for the first-order finite-temperature transition of the bosonic ring, Eq.~\eqref{eq:Ham0D}, obtained from the saddle point. The last row, $\nu=-1$, gives the fermionic values $-\beta_cUN$ for $\Nf=N$, using Eq.~\eqref{eq:NmN}. }
    \label{tab:CriticalFirstOrder0D}
\end{table}

\section{Frustrated hopping: bosons beyond color ferromagnetism}\label{sec:Frustration}
For unfrustrated hopping, the bosonic groundstate always lies in the sector $(\Nb)$, i.e., in the spinless problem.
In this section we therefore turn to frustrated hopping, for which the Eisenberg-Lieb theorem does not apply and the color structure of the groundstate becomes a genuine many-body question.
While few-site bosonic models at large $N$ are tied to fermionic ones by Sec.~\ref{sec:spectral}, this is no longer true in the conventional thermodynamic limit $L\to\infty$ at fixed $N$ and fixed density, allowing for distinct physics.
Indeed, sectors common to both statistics must fit in a box of $N$ rows and $N$ columns, whereas fermionic sectors at density $\Nf/L$ have $\alpha_1\le N$ and hence
$\ell(\alpha)\ge \Nf/N\to\infty$.
Throughout this section we denote by $\rho=\Nb/L$ the density per site.

We focus on the sawtooth chain and the kagome lattice, which both host a flat lowest band for frustrated hopping~\cite{Derzhko2015,Leykam2018}.
For spinless bosons ($N=1$), these lattices have been studied both in the hard-core limit, where the bosons are the magnons of spin-1/2 antiferromagnets close to saturation~\cite{Schnack2001,Schulenburg2002,Richter2004,DerzhkoRichter2004,Zhitomirsky2004}, and for soft-core bosons~\cite{Huber2010}.
There, the localized flat-band states can be occupied independently as long as they do not overlap, which leads to a close-packed crystal at $\rho=1/4$ (sawtooth) and $\rho=1/9$ (kagome)~\cite{Schulenburg2002,Zhitomirsky2004,Huber2010}.
For SU($N$) fermions, the same localized states underlie flat-band
ferromagnetism~\cite{Mielke_1991a,Mielke_1991b,Mielke1992,Tasaki1992,Mielke1993,Mielke1993b, LiuNieZhang2019, TamuraKatsura2019, jin2026flatbandferromagnetismsunhubbard}, and the finite-density problem maps onto a correlated percolation problem~\cite{Maksymenko2012,LiuNieZhang2019}.
To our knowledge, multi-component bosons with SU($N$)-symmetric interactions have not been studied in this frustrated flat-band setting.

\subsection{A flat-band theorem for bosons}\label{sec:flatbanth}
We write the kinetic term of Eq.~\eqref{eq:Ham} in its single-particle eigenbasis, $H_{\mathrm{kin}}=\sum_k \epsilon_k \tilde E_{k,k}$, and denote by
$\epsilon_{\min}=\min_k\epsilon_k$ the lowest single-particle energy and by $D_{\mathrm{flat}}$ the number of single-particle states at $\epsilon_{\min}$, not counting the color.

\begin{proposition}\label{prop:flatband}
If $U>0$, $\Nb\le D_{\mathrm{flat}}$ and $\Nb\le N$, the groundstate energy is
$E_0=\Nb\epsilon_{\min}+U\Nb/2$, and a groundstate lies in the color sector $(1^{\Nb})$, i.e., is totally antisymmetric in color.
For $\Nb=N$, it is an SU($N$) singlet.
\end{proposition}

\begin{proof}
Since $\tilde E_{k,k}\ge0$ and $\sum_k\tilde E_{k,k}=\Nb$,
$\langle H_{\mathrm{kin}}\rangle\ge\Nb\epsilon_{\min}$.
Since $E_{i,i}$ has non-negative integer eigenvalues, $E_{i,i}^2\ge E_{i,i}$ and
$\langle\frac U2\sum_iE_{i,i}^2\rangle\ge U\Nb/2$.
Both bounds are saturated by the state whose orbital part is the antisymmetrized product of $\Nb$ linearly independent lowest-band orbitals, which exists since $\Nb\le D_{\rm flat}$: its orbital diagram $(1^{\Nb})$ forbids double occupancy.
For bosons, this orbital sector is admissible if and only if $\Nb\le\min(N,L)$, and the color diagram is then also $(1^{\Nb})$.
\end{proof}

By Sec.~\ref{sec:dualities}, this bosonic flat-band antiferromagnetism is the equivalent of the flat-band ferromagnetism of SU($N$) fermions~\cite{Mielke_1991a,Tasaki1992,Mielke1993, LiuNieZhang2019, TamuraKatsura2019}.
The same orbital state is carried by the fully symmetric color irrep for fermions and by the fully antisymmetric one for bosons.
The key difference is that the fully antisymmetric color irrep exists only for $\Nb\le N$.

At finite density and fixed $N$, this simple proposition does not apply directly.
Nonetheless, for flat-band Hamiltonians with finite-range hoppings, it is possible to build compact localized states (CLS)~\cite{Sutherland1986, Bergman2008, Leykam2018, Flach_2014} within the flat band.
Through the examples of the sawtooth chain and the kagome lattice, we show how to generalize the proposition to build finite-density flat-band singlet-cluster states.

\subsection{The sawtooth chain}\label{sec:sawtooth}
\paragraph{Model.}
The sawtooth chain is the minimal lattice with a frustration-induced flat band.
For spinless bosons it hosts a crystal of localized states at $\rho=1/4$, and a supersolid above this density~\cite{Huber2010}; its fermionic SU($2$) and SU($N$) versions display flat-band ferromagnetism~\cite{Derzhko2010,TamuraKatsura2019}.
Cold-atoms realizations of this lattice have been proposed~\cite{Zhang2015}, and the flat band has been observed in photonic lattices~\cite{Weimann2016}.
We consider $N_u$ two-site unit cells, $L=2N_u$, with periodic boundary conditions:
\begin{multline}
H_{\mathrm{saw}} =t\sum_{j}\left(E_{(j,b),(j+1,b)}+\mathrm{h.c.}\right)
+\frac{U}{2}\sum_{j, s=a,b}E_{(j,s),(j,s)}^2  \label{eq:Hsaw}\\
+ t' \sum_j\left(E_{(j,b),(j,a)}+E_{(j,a),(j+1,b)}+\mathrm{h.c.}\right) ,
\end{multline}
where $(j,b)$ and $(j,a)$ denote the base and apex sites of cell $j$.
A gauge transformation on the apex sites flips the sign of $t'$, so that frustration is controlled by the sign of $t$ alone; we take $t>0$.
The single-particle energies are
\begin{equation}
    \epsilon_{k, \pm}=t\cos k \pm \sqrt{t^2 \cos^2k + 2t'^2(1 + \cos k)}.
\end{equation}

\paragraph{Flat-band cluster states up to $\rho_c=N/[2(N+1)]$.}
At $t' = \sqrt{2}\,t$, the lower band is exactly flat, $\epsilon_{k,-} = -2t$, and separated by a gap $2t$ from the upper band $\epsilon_{k,+}=2t(1+\cos k)$.
It is spanned by the CLS~\cite{Flach_2014}
\begin{equation}
    w_{j, \sigma}^\dagger = \frac{\sqrt{2}}{2}\, c^\dagger_{(j,b), \sigma}
    -\frac{1}{2}\, c^\dagger_{(j-1,a), \sigma}-\frac{1}{2}\, c^\dagger_{(j,a), \sigma},
\end{equation}
which are normalized, non-orthogonal but linearly independent~\cite{Derzhko2015}.
For a cluster of $N$ consecutive cells starting at $j_0$, the color-antisymmetric (unnormalized) state
\begin{equation}
\sum_{\pi\in S_N}\mathrm{sgn}(\pi)\prod_{k=0}^{N-1} w^\dagger_{j_0+k,\pi(k)}|0\rangle
\end{equation}
lies in the flat band and has no double occupancy: it saturates both bounds of Proposition~\ref{prop:flatband} and is an SU($N$) singlet.
Clusters separated by at least one empty cell share no site, and their product minimizes both the kinetic and the interaction energy.
For $\Nb=mN$, this construction yields exact singlet groundstates up to one cluster of $N$ cells per $N+1$ cells, i.e., up to the critical density
\begin{equation}
\rho_c(N)=\frac{N}{2(N+1)}.
\end{equation}
At the critical density, the groundstate is therefore a crystal of singlet $N$-mers separated by single empty unit cells.
For $N=1$, it reduces to the known crystal of localized states at $\rho=1/4$
\cite{Richter2004,DerzhkoRichter2004,Huber2010}.
For $N=2$ the clusters are singlets of two neighboring localized orbitals, and the close-packed state is a crystal of such singlets, reminiscent of the exact valence-bond groundstates of the spin-1/2 sawtooth chain~\cite{MontiSuto1991,NakamuraKubo1996,Sen1996} and of the Majumdar-Ghosh chain~\cite{MajumdarGhosh1969}. 
For general $N$ it is the analogue of SU($N$) valence-bond solids built from $N$-site singlets~\cite{AffleckAMR1991,GreiterRachel2007}.
By contrast, the fermionic clusters built on the same orbital states are color-symmetric and can be arbitrarily large~\cite{Maksymenko2012,LiuNieZhang2019}.\\

We present in Fig.~\ref{fig:EnergiesSawtooth} the energy per particle of the lowest state in each color sector $\alpha$, for a periodic sawtooth chain with $L=16$ sites ($N_u=8$ unit cells), $U=1$, and $\Nb=5$ to $8$ particles.
For each irrep $\alpha$, the most compact flat-band state that can be constructed is a set of $\T\alpha_j$ consecutive CLS $w_{j, \sigma}$, separated by a single empty unit cell.
The lowest state of the irrep $\alpha$ thus reaches the flat-band energy if and only if $\Nb+\alpha_1\le N_u$, with the single exception of $\alpha=(1^{N_u})$, where the cluster covers the whole chain and closes around the ring.
For $\Nb\le4$, i.e., at or below $\rho_c(1)=1/4$, all irreps reach the flat-band energy.
For $M=5$, the density exceeds $\rho_c(1)$: spinless bosons are necessarily frustrated, and the irreps $(5)$ and $(4,1)$ lie above the flat-band energy.
For $\Nb\ge2$, flat-band states exist in the irreps with $\alpha_1\le3$, e.g.\ $(3,2)$.
Similarly, at $\Nb=6$, which is precisely $\rho_c(3)$, flat-band states exist only for $\alpha_1\le2$, hence only for $N\ge3$, e.g.\ in $(2,2,2)$.
At $\Nb=7$, which is precisely $\rho_c(7)=7/16$, only $(1^7)$ reaches the flat-band energy, which requires $N\ge7$.
Finally, at $M=8$ ($\rho=1/2$), the density exceeds $\rho_c(N)$ for all $N$.
The flat-band state found in $(1^8)$ for $N\ge 8$ is a finite-size effect: it is the completely filled flat band.

Above the critical density, the interaction energy no longer vanishes, and antisymmetric orbital wave functions are favored, as they suppress multiple occupancies.
Consistently, we observed that at fixed $\ell(\alpha)$, the groundstate generally lies in the most balanced irrep possible.
We verified this for $U\in[0.1,20]$ and $\Nb = 6$ and $8$.
For fixed $N$, it is also generally the case: the groundstate tends to be in the most antisymmetric irrep (with largest $\ell(\alpha)$).
This rule is, however, not general.
For $\Nb=8$, the irrep $(2^4)$ remains the groundstate for SU($5$) and SU($6$), below the more antisymmetric admissible irreps such as $(2,2,2,1,1)$.
We nonetheless expect this to be a finite-size effect, caused by the absence of SU($5$) and SU($6$) singlets at this particle number.
At half filling $\Nb=N_u$, the first size at which complex SU(3), SU(4) and SU(6) singlets coexist is $\Nb=12$, where $(2^6)$ and $(3^4)$ have dimensions $3.4\times10^{9}$ and $2.7\times10^{10}$.
The latter lies beyond the reach of this work.

\begin{figure*}
    \centering   \includegraphics[width=\linewidth]{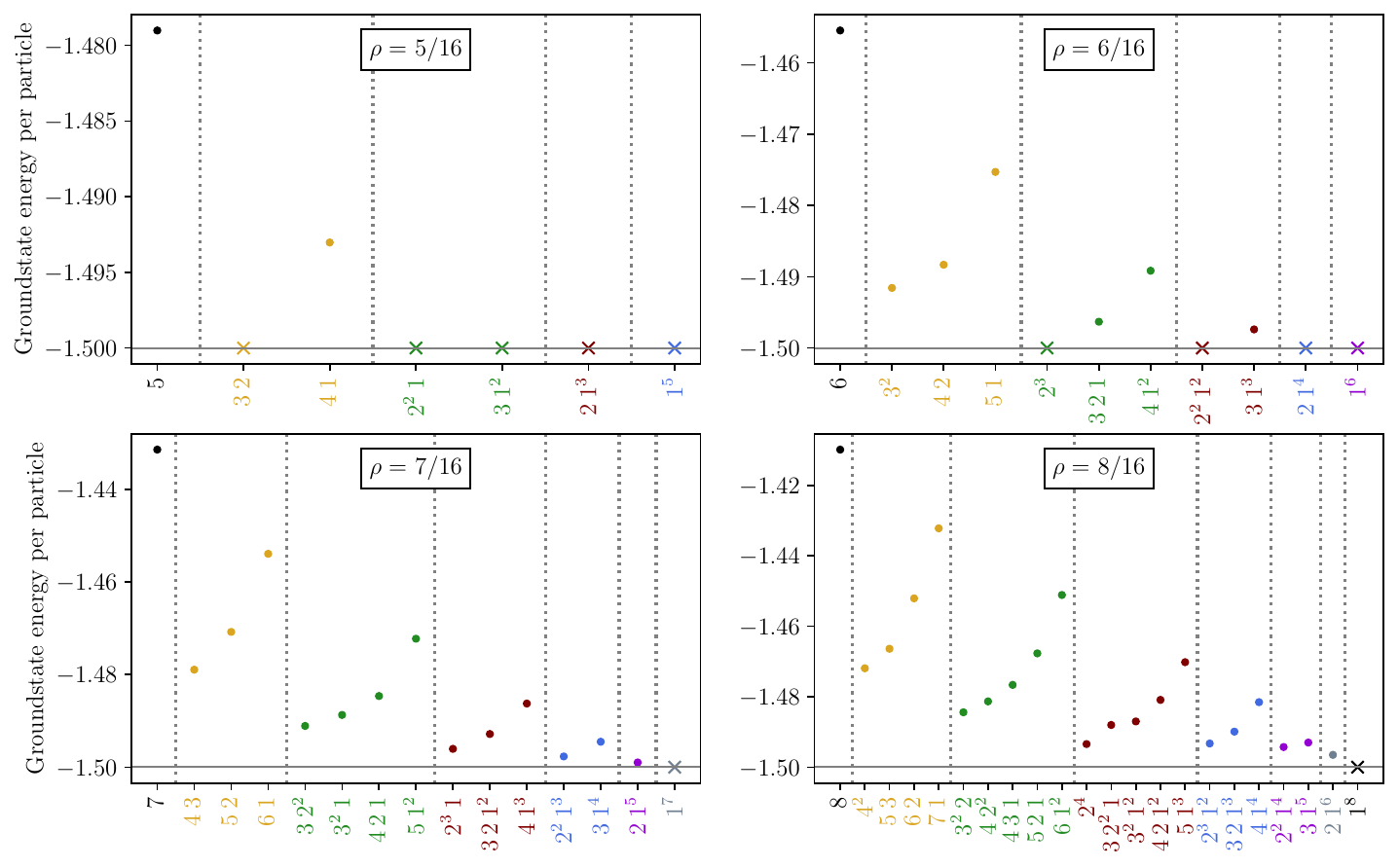}
    \caption{Energy of the lowest state of each irrep for a sawtooth chain with $L = 16$ sites at $U=1$, for different particle densities. The horizontal lines mark the flat-band energy $-2t + \frac{U}{2}$, with crosses lying exactly at this energy. Vertical lines and colors are guides to the eye and separate irreps with increasing $\ell(\alpha)$. As we increase densities towards $\rho = 1/2$, a larger symmetry is required to stabilize the exact antiferromagnetic flat-band states.}
    \label{fig:EnergiesSawtooth}
\end{figure*}

\subsection{The kagome lattice} \label{sec:kagome}
\paragraph{Model.}
The kagome lattice is the two-dimensional counterpart of the previous lattice, with the bosonic Hamiltonian
\begin{equation}
    H_{\rm kag} = t\sum_{\langle i, j\rangle} (E_{i, j} + \mathrm{h.c.})
    + \frac{U}{2} \sum_{i} E_{i, i}^2,
\end{equation}
with $t > 0$ and periodic boundary conditions.
We denote by $N_u$ the number of unit cell in the lattice.
The lowest single-particle band is flat, at $\epsilon=-2t$, with a quadratic band touching at $\vec{k} = \vec{0}$.
It is spanned by the $N_u - 1$ hexagon states
\begin{equation}
    w^\dagger_{\hexagon, \sigma} =\frac{1}{\sqrt{6}} \sum_{j = 1}^6 (-1)^j c^\dagger_{j, \sigma},
\end{equation}
where $j$ labels the sites around a hexagon, together with two additional non-contractible loop states winding around the two cycles of the torus~\cite{Derzhko2015}.
Flat bands of the kagome lattice have been realized with ultracold bosons~\cite{Jo2012,Leun2020}, where the flat band is accessed as an excited band. 
Making it the lowest band requires reversing the sign of the tunneling, e.g.\ by periodic driving~\cite{Lignier2007,Struck2011,Eckardt2017} or through negative temperature states~\cite{Braun2013}.\\

\paragraph{Flat-band cluster states.}
For SU(2) fermions, flat-band ferromagnetism on the kagome lattice goes back to Refs.~\cite{Mielke_1991a, Mielke_1991b, Mielke1992,TanakaUeda2003}, and it was recently extended to SU($N$)~\cite{jin2026flatbandferromagnetismsunhubbard}.
For bosons, the constraint $\ell(\alpha)\le N$ changes the structure of the groundstates.
For $\Nb\le N$, Proposition~\ref{prop:flatband} applies.
At finite density, we use the construction of Sec.~\ref{sec:sawtooth}: $n\le N$ neighboring hexagon states in a fully color-antisymmetric state minimize the energy for $\Nb=n$, and the symmetrized product of clusters that share no site, i.e., built on non-adjacent hexagons,  is a groundstate at the corresponding density.
If all clusters contain $n=N$ hexagons, it is an SU($N$) singlet.

The maximal density reachable with this construction is the solution of a packing problem on the triangular lattice formed by the hexagon centers: find the densest arrangement of non-adjacent connected clusters of at most $N$ hexagons.
The first version is the maximum $N$-order bounded component set problem~\cite{mishra2018complexitymaximumkorderbounded}, which reduces to the maximum independent set for $N=1$~\cite{Karp1972} and the maximum dissociation set for $N=2$~\cite{Yannakakis1981}; the singlet version is an induced packing of connected $N$-vertex subgraphs.
These problems are NP-hard on general graphs, but we are not aware of exact results on the triangular lattice for $N\ge2$.
For $N=1$, the problem is the close packing of hard hexagons~\cite{Baxter1980}: one out of three triangular sites can be occupied, giving $\rho_c=1/9$~\cite{Schulenburg2002, Zhitomirsky2004}.
Solving the packing problem on periodic cells of up to $54$ hexagon centers, we obtain the following lower bounds on the singlet critical densities:
\begin{equation}
\rho_{c, 2} \geq \frac{2}{15},\qquad \rho_{c, 3}\geq \frac{4}{27},\qquad \rho_{c, 4} \geq \frac{1}{6}.
\end{equation}
The corresponding packings are shown in Fig.~\ref{fig:patterns}.
A detailed study is left for future work.\\

\begin{figure}
\vspace{0.4cm}
    \centering
    \includegraphics[width=0.9\linewidth]{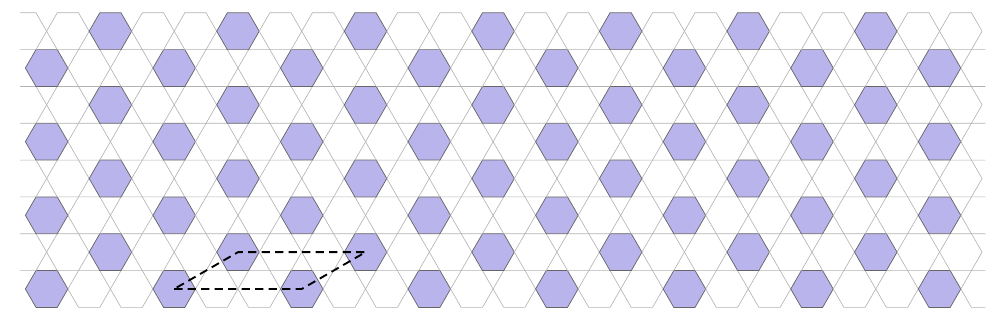}\\
    \includegraphics[width=0.9\linewidth]{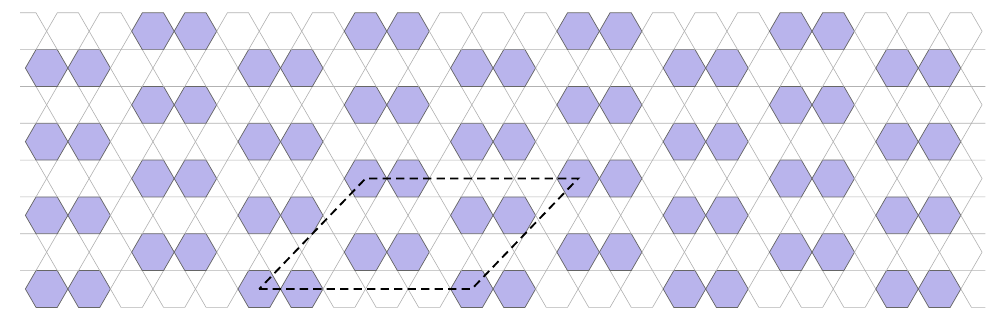}\\
    \includegraphics[width=0.9\linewidth]{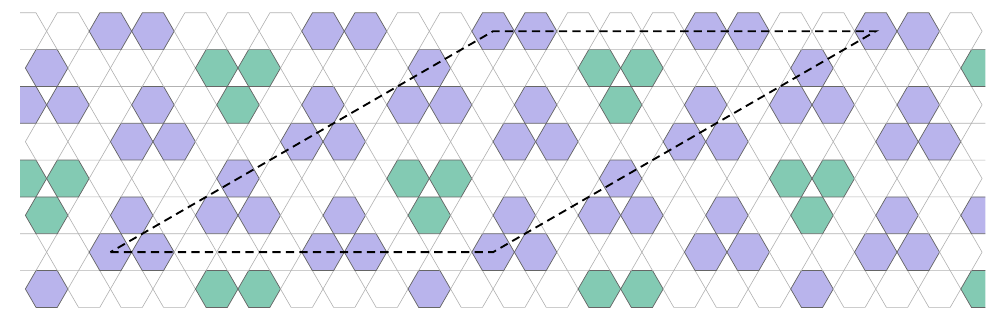}\\
    \includegraphics[width=0.9\linewidth]{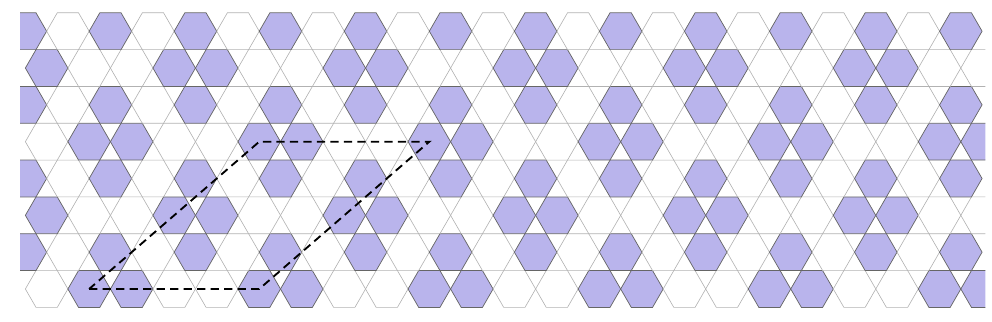}\\
    \caption{Packing solutions for the $N$-cluster problem ($N = 1$ to $4$), with corresponding densities (per site) $\frac{1}{9}$,  $\frac{2}{15}$,  $\frac{4}{27}$ and  $\frac{1}{6}$.}
    \label{fig:patterns}
\end{figure}

Above the critical density, flat-band states with minimal interaction energy no longer exist, and the groundstate color irrep results from a competition between kinetic and interaction energies.
We show in Fig.~\ref{fig:kagomeenergies} the lowest energy in each color sector for the kagome clusters with $L=9$ and $L=12$ sites.
At integer filling ($L=\Nb=9$), strong interactions favor ferromagnetism: at $U=10$, the groundstate lies in the fully symmetric irrep $(\Nb)$, and at fixed $\ell(\alpha)$ the energy decreases with increasing $\alpha_1$.
This follows from the Mott limit.
At second order in $t/U$, the effective Hamiltonian is $H_{\rm eff}=-(2t^2/U)\sum_{\langle ij\rangle}(1+P^{\mathrm{orb}}_{ij})$, where $P^{\mathrm{orb}}_{ij}$ exchange the orbital states at position $i$ and $j$,  whose groundstate lies in the orbital sector $(\Nb)$.
Note that the correspondence of Sec.~\ref{sec:spectral} also guarantees ferromagnetism as the Mott limit is antiferromagnetic for fermions.
Since $H_{\rm eff}$ depends on $t^2$ only, the sign of the hopping, and hence the frustration, drops out at this order.
Although the Eisenberg-Lieb theorem of Sec.~\ref{sec:gain} does not apply to the kagome lattice, strong interactions thus restore an SU($N$) ferromagnetic groundstate.

At weaker interactions ($U=1$), the kinetic energy of the flat band, the dispersive band touching it at $\Gamma$, and the interaction energy compete, and the groundstate lies in intermediate irreps.
At fixed $\ell(\alpha)$, the more antisymmetric irreps remain favored, as they reduce multiple occupancies.
Away from integer filling, even at strong coupling ($U = 100$, $\rho=2/3$), the ferromagnetic order is not selected, consistent with the Mott-picture.
Instead, antiferromagnetic singlet states at intermediate $N$ are stabilized.
We stress that these clusters accommodate a single flat-band cluster at most, so that the intermediate-coupling results are strongly affected by finite-size effects, and should not be taken as indicative of thermodynamic phases.

\begin{figure*}
    \centering
    \includegraphics[width=0.98\linewidth]{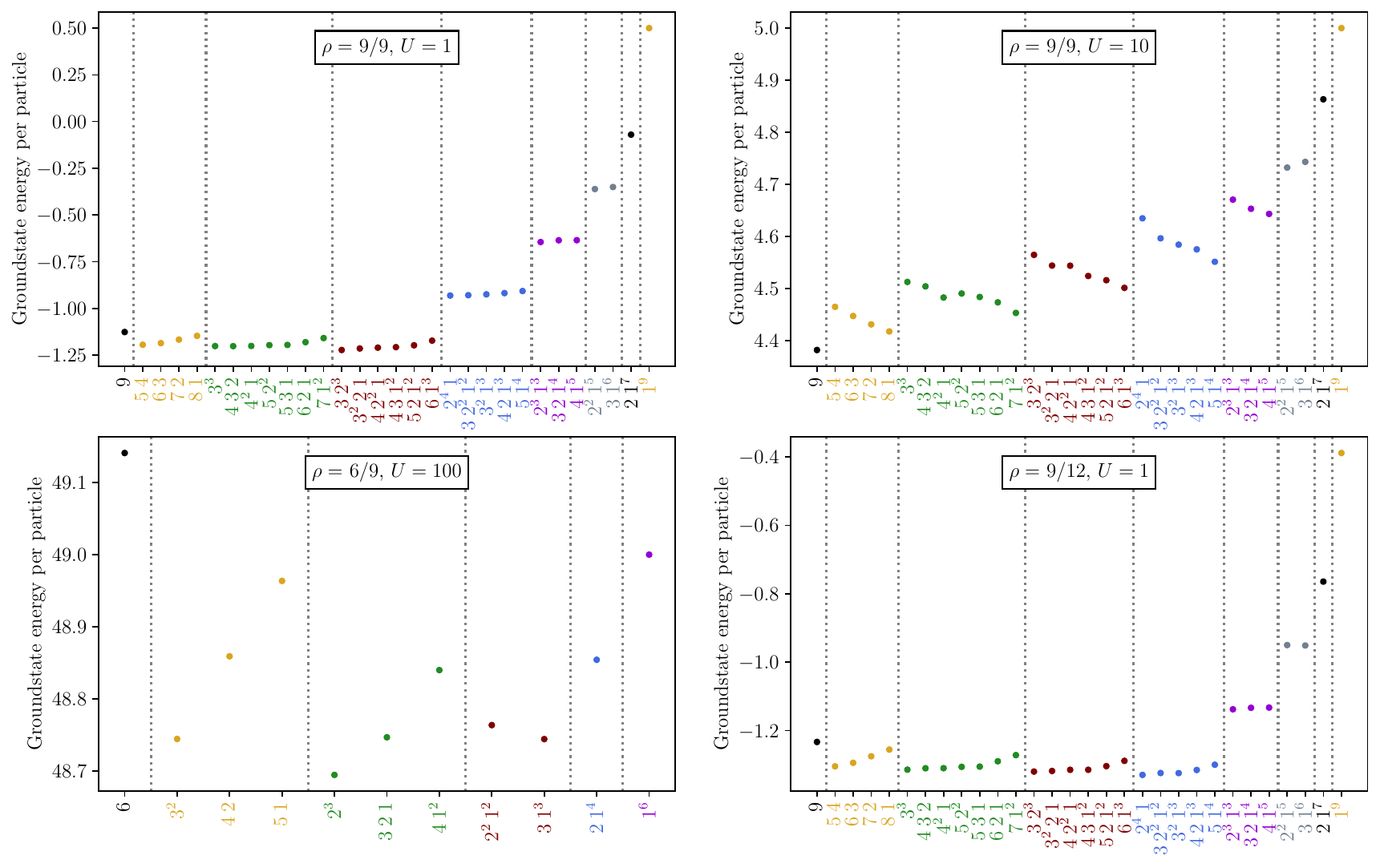}
    \caption{Groundstate energy per particle in each symmetry sector for $L = 9$ and $L = 12$, at several densities and interaction strengths. At integer densities, a ferromagnetic groundstate is stabilized in the Mott regime. At densities between integer and the critical packing density, intermediate singlet states are stable.}
\label{fig:kagomeenergies}
\end{figure*}

\section{Conclusion}
We have extended color-resolved exact diagonalization from SU($N$) fermions to SU($N$) bosons, providing an exact numerical tool for SU($N$)-symmetric bosonic gases~\citep{Mukherjee2025,MukherjeeHutson2025}.
The method is based on the Howe duality between U($L$) and U($N$).
Orbital and color irreps coincide for bosons while they are transposed for fermions.
This has three general consequences: the orbital spectrum in each sector is independent of $N$ or of the statistics, bosonic and fermionic spectra obey explicit inclusions on any lattice, and the canonical partition functions satisfy the exact correspondence $Z_{\rm b}(-N)=(-1)^{\Nb}Z_{\rm f}(N)$.
For non-frustrated hoppings whose signs can be gauged away, the Eisenberg-Lieb~\cite{Suto_1993,Eisenberg2002} theorem makes the bosonic groundstate an SU($N$) ferromagnet, and the color-resolved basis is then useful for excited sectors, finite temperature and dynamics.

For few-site rings at large $N$, we obtained the finite-temperature transitions of attractive bosons at arbitrary filling $\nu=\Nb/N$.
In the strong-coupling limit, the transition is continuous for $L=2$, at $\beta_cU\Nb=-4/(2+\nu)$, and first order for $L\ge3$; the hopping shifts $T_c$ at finite $|U|\Nb/t$.
The fermionic results of Ref.~\cite{Herviou2026} are recovered for $\nu=-1$, a direct illustration of the $N\to-N$ correspondence.

With frustrated hopping, bosons can escape color ferromagnetism.
We proved that for $\Nb\le\min(N,D_{\rm flat})$ a flat lowest band stabilizes a fully color-antisymmetric groundstate, which is an SU($N$) singlet for $\Nb=N$.
At finite density, exact groundstates can be built from clusters of at most $N$ overlapping  CLS, with distinct clusters sharing no site.
On the sawtooth chain, the lowest state of the irrep $\alpha$ reaches the flat-band energy if and only if $\Nb+\alpha_1\le N_u$.
It yields the critical density $\rho_c(N)=N/[2(N+1)]$ and generalizes the localized-magnon crystal of spinless bosons.
On the kagome lattice, the same question maps onto a packing problem on the triangular lattice of hexagon centers, with connected components of at most $N$ hexagons.

Above the critical density, the color structure of the groundstate results from a genuine competition between kinetic and interaction energies.
On the sawtooth chain at half filling, antisymmetric irreps remain generally favored, as they suppress multiple occupancies.
On kagome clusters at integer filling, strong interactions restore an SU($N$) ferromagnet.
At intermediate couplings, the groundstate lies in intermediate irreps, a regime where our clusters remain strongly affected by finite-size effects.\\

The present numerical results remain limited by exact diagonalization and by the intrinsically large Hilbert spaces of bosonic systems.
Extending the SU($N$)-resolved DMRG algorithm of Ref.~\cite{Nataf2024DMRG} to bosons, with a cutoff on the local occupation, is a natural next step and would give access to significantly larger one-dimensional systems.
Two-dimensional systems remain challenging despite the exact resolution of the SU($N$) symmetry.
For flat-band physics, projecting the orbital problem onto the flat band reduces its dimension from $d_{\alpha,L}$ to $d_{\alpha,D_{\rm flat}}$ and should allow exact tests of the critical densities on larger tori.
Beyond this regime, SU($N$)-symmetric two-dimensional tensor networks appear the most promising route to strongly interacting frustrated bosons, while quantum Monte Carlo, which suffers from a sign problem with frustrated hopping, remains limited to unfrustrated lattices.

\acknowledgments
The author thanks Jeanne Colbois, Philippe Lecheminant and in particular Pierre Nataf for useful discussions.
He acknowledges the support of the ANR JCJC ANR-25-CE30-2205-01.
During the preparation of this work, the author used generative AI to improve the language and readability of the manuscript.
After using this tool, the author reviewed and edited the content as needed and takes full responsibility for the content of the publication.

\bibliography{bibliography_additions}

\end{document}